\documentclass[11pt]{article}
\usepackage[colorlinks=true, citecolor=blue]{hyperref}
\usepackage{enumerate,tikz}
\usepackage{mathtools, amssymb, latexsym, amsmath, amsfonts, amsthm, xcolor,braket}
\usepackage{array, fancyhdr, bm, listings, lilyglyphs,musicography}
\usepackage{subcaption} 
\usepackage{thmtools, thm-restate,physics}
\declaretheorem{theorem}
\DeclareMathOperator{\bw}{\mathcal{W}} 

\usepackage{ bbold }

\theoremstyle{plain}
\newtheorem{lemma}[theorem]{Lemma}

\newtheorem{corollary}[theorem]{Corollary}

\newtheorem{fact}{Fact}

\theoremstyle{definition}

\DeclareMathOperator{\col}{col}
\DeclareMathOperator{\sign}{sgn}



\newcommand\sym{\mathcal{S}_n}

\providecommand{\keywords}[1]{\textit{Keywords:} #1}
\title{Non-uniform Mixing of Quantum Walks on the Symmetric Group}

\date{}

\author{Avah Banerjee}
\newcommand{\FormatAuthor}[3]{
\begin{tabular}{c}
#1 \\ {\small\texttt{#2}} \\ {\small #3}
\end{tabular}
}
\author{
\begin{tabular}[h!]{lcr}
   \FormatAuthor{Avah Banerjee \thanks{AB is supported by NSF award no. CCF-2246144.} }{banerjeeav@mst.edu}{Missouri S\&T}
\end{tabular}
}

\begin{document}

\maketitle

\begin{abstract}
It is well-known that classical random walks on regular graphs converge to the uniform distribution. Quantum walks, in their various forms, are quantizations of their corresponding classical random walk processes. Gerhardt and Watrous (2003) demonstrated that continuous-time quantum walks do not converge to the uniform distribution on certain Cayley graphs of the Symmetric group, which by definition are all regular. In this paper, we demonstrate that discrete-time quantum walks, in the sense of quantized Markov chains as introduced by Szegedy (2004), also do not converge to the uniform distribution. We analyze the spectra of the Szegedy walk operators using the representation theory of the symmetric group. In the discrete setting, the analysis is complicated by the fact that we work within a Hilbert space of a higher dimension than the continuous case, spanned by pairs of vertices. Our techniques are general, and we believe they can be applied to derive similar analytical results for other non-commutative groups using the characters of their irreducible representation.

\vspace{0.5cm}
\keywords{Quantum Walks, \and Symmetric Group, \and Non-commutative Fourier analysis}\\

\textbf{Mathematics Subject Classification:} 81P68, 20C30.

\end{abstract}

\section{Introduction}
The phenomenon of random walks on graphs has been widely studied and holds significant applications across a myriad of problems in computational sciences. They have been instrumental in developing randomized and approximation algorithms \cite{lovasz1993random}. 
Random walks can be characterized by Markov chains and they can be fully characterized using methods from spectral graph theory.

We look at the problem of sampling from the symmetric group via a quantization of random walk. The study of sampling from groups has a rich history \cite{diaconis1988group, diaconis1981generating,babai1991local}. Particularly, sampling an element from the symmetric group has been well-studied in the classical setting \cite{diaconis1981generating}, especially with respect to functions over elements of groups. Sampling from group elements ties with certain random walks; in some cases, even if the original sampling problem does not involve sampling from a group element (for example, the famous Ehrenfest process). These random walks take place on the Cayley graphs of the groups, which are constructed using some generating set. Since Cayley graphs are regular, a uniform random walk on them converges to the uniform distribution. However, this does not seem to be the case in the quantum setting. We extend the analysis of Gerhardt and Watrous (2003) \cite{gerhardt2003continuous} to demonstrate the uniformity of the distribution, both instantaneous and average, arising from the quantization of a uniform Markov chain on the symmetric group.

Unlike \emph{classical} random walks a \emph{quantum walk} propagates using the principle of quantum mechanics. Few difference of note include - 1) Instead of real probabilities the state of the walk is specified by complex probability amplitudes\footnote{However, in some case if the amplitudes are constrained to be in $\mathbb{R}$, working with them becomes slightly simpler.}. 2) The random (walk) coin is now replaced by a unitary transformation. The unitary evolution ensures the walk is reversible\footnote{For open systems the walk operator need not be unitary. Interspersing walking with measurements also leads to non-unitary dynamics\cite{kendon2007decoherence}}.
2) Propagation of the walk generates a superposition state overs all possible positions available to the walker. 
3) Finally, we can sample the positions by applying suitable measurements on the state of the walker.

There are various (somewhat equivalent) models of quantum walks.
Study of quantum walks has a long history, going back to the early works of Feynman, Meyer, Aharonov, Gutmann and others \cite{aharonov1993quantum, farhi1998quantum, meyer1996quantum}.
The hope is that quantum walk can emulate the success of random walk in the development of classical algorithms in developing quantum algorithms.  Quantum or classical walk\footnote{Henceforth we will refer to classical random walk simply as classical walk.} has been primarily used as a generative models for probability distributions. 
Hence, two of the most important properties to study are the kind of distributions they can generate and their converging behavior. In general, quantum walks do not converge to a stationary distribution. However, their time-averaged distribution (introduced later) does converge.
Quantum walk has been shown to generalize Grover's diffusion based search on graphs. It has been used to obtain currently best known quantum algorithms for certain problems. Most notable among them are element distinctness, triangle finding, faster simulation of Markov chains, expansion testing etc. \cite{magniez2011search, ambainis2008quantum, apers2020expansion}. 

\subsection{Overview of our techniques and Results.}
In this paper, we focus on a discrete-time model of quantum walk. The model we examine has its origins in the seminal paper by Aharonov et al.\cite{aharonov2001quantum}. Since its introduction, numerous variants of discrete-time quantum walks (DTQW) have emerged. When it comes to accelerating randomized algorithms by harnessing the faster mixing properties of quantum walks, the most extensively studied framework involves the quantization of a classical Markov chain. This framework was first introduced by Szegedy \cite{szegedy2004quantum}, and has since been expanded upon and applied to a multitude of graph search algorithms within the black-box query model. In this paper, we employ the Szegedy quantum walk framework to analyze the distribution properties of a specific type of Cayley graph of the symmetric group. Our primary concern is not the mixing time or other measures of convergence, but rather illustrating how the probability distribution deviates from that of the corresponding classical Markov chain. This research implies that the probability of observing a group element is intricately tied to the ``weight" of its various irreducible representations. In the case of abelian groups, given that their representations are 1-dimensional, they play a consistent role for all group elements, and, as demonstrated previously, such walks converge to the uniform distribution. This distinctive difference renders the study of such walks for symmetric and other simple non-abelian groups considerably more intricate. Further discussions on previous results can be found in Section~\ref{sec: prev and releted}.

Given a Markov chain with its transition matrix $P$, we begin by constructing a bipartite walk on the combined state space $X \times X$. The transition matrix of this bipartite walk forms the basis for deriving a unitary operator in the quantum context. Informally, for each state $x \in X$, one constructs a vector $\vec{\phi}_x$ that represents a superposition of the edges linking $x$ to its neighbors, with weights corresponding to the transition amplitudes. These vectors are utilized to define a reflection operator as well as a shift operator (which will be introduced later). The reflection operator allows the quantum walk to ``propagate in superposition" along the edges adjacent to a vertex. The shift operator alternates the propagation direction from left-to-right and vice versa.
The composition of these operators results in a unitary $\bw$, which characterizes a step of the quantum walk. 
This construction facilitates a relatively straightforward determination of the spectral decomposition of $\bw$ in relation to the spectra of $P$. The elements of $X$ can be interpreted as vertices of an edge-weighted directed graph, where the weights correspond to the transition probabilities. In our context, this graph is associated with a specific Cayley graph of the symmetric group. The quantization of the bipartite Markov chain gives rise to a quantum walk, which fundamentally occurs on the edges of the original graph. Due to the inverse closure of the generating set we utilized, this graph is undirected.

In the case of continuous-time quantum walks, which evolve based on the Hamiltonian $e^{itL}$ where $L$ is the graph Laplacian, Gerhart and Watrous studied the walk on several Cayley graphs of the symmetric group. They utilized the spectral decomposition of the random walk operator in terms of the irreducible representations (irreps) of the symmetric group, initially derived by Diaconis \cite{diaconis1988group}, to determine the probability of observing an $n$-cycle for the quantum walk. They demonstrated that this probability, $O(2^{-2n}/(n+1)!)$, is exponentially smaller than in the case of the uniform distribution (which is $1/n$). We extend their technique in the setting of the Szegedy walk. We apply the spectral decomposition of $P$ in terms of the irreps of the group to construct a similar, albeit more technical, spectral decomposition of $\bw$. This enables us to similarly upper bound the probability of observing an $n$-cycle, in our case determined by all edges incident to $n$-cycles. However, due to the difference in the Hilbert space of the quantized Markov chain compared to the continuous-time version, our analysis presents a considerably greater challenge.

The analysis highly depends on the tractability of working directly with irreps. Unfortunately, the irreps are matrices with no simple formulaic description. This restricts us to focusing our analysis on cases where we can use the characters of the group elements instead of the irreps. In light of this, we limit our study to generating sets that are conjugacy closed. Specifically, in this paper, we focus on the generating sets which consist of transpositions. Even with this limitation, the analysis is still influenced by the choice of the initial state and the form of the final state. Particularly, the support of the probabilities in the final state should also be over a conjugacy-invariant set. However, this isn't an issue for us as, when testing the probability of observing an $n$-cycle, we can choose the uniform distribution over all $n$-cycles (more technically, all edges incident to $n$-cycles) and determine its overlap with the final state. Much of the technical calculations in this paper are focused on determining this overlap using the characters of the group. This lead to our main theorem.

\begin{restatable}{theorem}{maintheorem}
For any constant $\beta > \frac{81}{16}$,
\begin{align*} 
\norm{\bra{\phi_{[n]}}\bw^t\ket{\phi_{\mathbb{e}}}} = O\left(\frac{n^{20}\beta^{2n}}{n!}\right),
\end{align*}
where $\ket{\psi_{[n]}} = \frac{1}{\sqrt{(n-1)!d}}\sum_{g \in [n], s \in S}\ket{g, gs}$ is the uniform superposition of all edges incident to $n$-cycles, and $\ket{\psi_{\mathbb{e}}} = \frac{1}{\sqrt{d}}\sum_{s \in S}\ket{\mathbb{e}, s}$ is the uniform superposition of all edges incident to the identity permutation.
\label{thm: main theorem}
\end{restatable}
Ignoring the polynomial factor, which arises due to the technical limitations of our approach, we observe that the walk operator behaves roughly similarly to that of the continuous version. 

\subsection{Discussion}
\subsubsection{Classical Complexity of Generating Quantum Walk Distribution, Localization}
Consider a unitary operator $U$ drawn from the unitary group $\mathbb{U}(n)$ according to the Haar measure. When $U$ acts on the state $\ket{0^n}$, it results in the state $\ket{\psi}$. It's worth noting that $\ket{0^n}$ can be replaced with any fixed initial state, not necessarily the all-zero state. Let $p(x)$ represent the probability of observing the state $\ket{x}$ when the system is in state $\ket{\psi}$, with $x$ being a computational basis state. 

It is a well-established fact that when $p(x) \in [0,1]$ is viewed as a continuous random variable (determined by the Haar measure on $\mathbb{U}(n)$), the distribution of probability values, $p$, conforms to the Porter-Thomas distribution as given by: $\mathbb{P}[p] = 2^ne^{-2^np}$.
An intrinsic property of this distribution is that the probability amplitude distribution of $\ket{\psi}$ is anti-concentrated; this means that no specific basis state is noticeably more or less probable than the uniformly distributed probability value of $2^{-n}$. Nevertheless, there isn't a classical process available to effectively approximate such outputs, even though it's feasible to devise a classical stochastic process that can mimic the Porter-Thomas distribution for these probabilities \cite{barak2020spoofing,bouland2018quantum}.
Our result, which is consistent with those from continuous-time cases, reveals that the unitary operator of the Discrete-Time Quantum Walk (DTQW) on the symmetric group is considerably different from a Haar-random unitary. Additionally, its probability distribution deviates markedly from uniformity. We observe that the quantum walker has certain blind spots, which may hint at some localization phenomenon and consequently a lack of anti-concentration. It is not \textit{a priori} evident that such distributions can be generated efficiently in the classical setting.




\subsubsection{Beyond Class Functions}
While the techniques presented here are farely general, their capacity to derive analytical results largely hinges on the ability to use characters of certain irreducible representations of the symmetric group. Generally, the spectrum of the DQTW operator will depend directly on the irreps, unless the discriminant matrix, derived from the Markov chain, possesses some special structure. As outlined in \cite{diaconis1988group}, the eigenvalue expression utilizing class functions is valid for any matrix whose $[g,h]^{th}$ entry can be expressed by a class function $f(g^{-1}h)$. More comprehensively, for any transition matrix $P$, one can form a block-diagonal decomposition:
\begin{align*}
P = \varphi M \varphi^{\dagger}
\end{align*}
Here, the columns of the $|G|\times |G|$ matrix $\varphi$ form an orthonormal set of vectors spanning the vector space defined by elements of $G$. Moreover, $M$ is a block-diagonal matrix, with blocks corresponding to components of the Fourier transformation of $P$ (defined as a function from $G$ to $[0,1]$ where $P(g,h) = f(g^{-1}h)$ for a certain probability distribution $f$ on $G$, not necessarily a class function) in terms of the irreps of $G$. In scenarios where $f$ is a class function, $M$ is strictly a diagonal matrix, and the column vectors of $\varphi$ are the vectors $\rho_{\mu,i,j}(g)$, with $\rho_\mu$ denoting an irrep of $G$. Nonetheless, within this broader framework, the spectral decomposition of the Szegedy walk operator is directly influenced by the matrix entries of the irreps, not solely the trace (commonly known as the characters of the irreps). Currently, there is no apparent method to broaden our analysis beyond class functions.


\subsubsection{Discrete Heisenberg Group}
The discrete Heisenberg group has found applications in physics \cite{floratos2016discrete,aliferis2017sl} as well in complexity theory \cite{lee2006lp}. It is one of the simplest non-abelian extension of the regular 2d lattice, for which random walks have been thoroughly studied. Furthermore, it has an elegant description with respect to its \emph{center}.
The 3-dimensional discrete Heisenberg group $H_3(n)$ over $\mathbb{Z}/n\mathbb{Z}$ is defined by the following multiplication rule : $(x,y,z)(x',y'z') \to (x+x',y+y', z+z'+xy')$ (modulo $n$).
Dynamics of random walks over them are well understood, and known to converge to the uniform distribution in $O(n^2)$ steps \cite{bump2017exercise,zack1990measuring}. Here, we are especially interested in the case where $n = p$, a prime. This group is \emph{extra-special}, in particular, $H_3(p)/ Z$ is abelian and $Z$ is cyclic where $Z$ is the center of $H_3(p)$. For this case we can to consider the \emph{Schreier coset graph} of $H_3(p)$ with respect to the cosets of $Z$. As far as we are aware, quantum walks have not been studied on coset graphs in the discrete setting (for an example in the continuous case see \cite{osborne2004quantum}). 
The techniques presented here could be applied to this group, possibly taking advantage of its ``nested" structure and more manageable representations.

\section{Preliminaries}
\subsection{Symmetric Group, Cayley Graphs and Representation Theory}
\ifx false
\begin{figure}[h]
	\includegraphics[width=5cm]{Figures/gamma-4.eps}
	\centering
	\caption{The graph $\Gamma_4$. Edges corresponding to the element $(12)$ (resp. $(1\cdots n)$) are colored green (resp. blue). It has $24$ vertices and has a diameter of $6$.}
\label{fig: gamma n}
\end{figure}
\fi
\subsubsection{Cayley Graphs}
Let $(G, \circ)$ be any finite group, and let $S$ be a generator of $G$. We define $|G| = N$ and $|S| = d$.
The Cayley graph of the pair, denoted as $\Gamma(G, S)$, is a directed graph $\Gamma$, defined as follows:
The vertex set is $V(\Gamma) = G$, and the edge set is defined as
\begin{align*}
    E(\Gamma) &= \{(g,h) \mid g, h \in G, \exists s \in S \text{ such that } g^{-1}\circ h \in S\}.
\end{align*}
Henceforth, we omit the ``$\circ$" and simply write $g\circ h$ as $gh$ for all $g, h \in G$.
If $S$ is closed under taking inverses, i.e., $s \in S \implies s^{-1} \in S$, then $\Gamma$ is undirected. 
In this paper, we set $G = \sym$, the symmetric group of permutations of $n$ elements.
We use $\mathbb{e}$ to denote the identity permutation, where $\mathbb{e} = (1)(2)\cdots (n)$.
We will exclusively work with the generating set composed of all transpositions in $\sym$, i.e., $S = \{(i,j) \mid i \ne j, i,j \in [n] \}$. Thus, $S$ is closed under conjugation, and $\Gamma_n = (\sym, S)$ is a ${n \choose 2}$-regular undirected graph. Throughout this paper, we use \(g\), \(h\), \(x\), \(y\), etc. to denote generic group elements. Occasionally, we use \(\pi\) and \(\sigma\) to emphasize elements of the symmetric group.

 The quantum walk studied in this paper does not take place directly on $\Gamma_n$ but on a bipartite extension, denoted as $\Gamma^{\musDoubleFlat}_n$, of $\Gamma_n$. Where 
\begin{align*}
   \Gamma^{\musDoubleFlat}_n = (\sym \times \sym, \{{\pi, \sigma}\mid \pi, \sigma  \in \sym \text{ and } \exists \tau \in S \text{ such that } \pi = \sigma \tau\})
\end{align*}
We further elaborate on this when introducing Szegedy walks in Section~\ref{sec: szegedy walk}.


\subsubsection{Representation of the Symmetric Group}
\ifx false
\begin{figure}[h]
	\includegraphics[width=5cm]{Figures/fig1.eps}
	\centering
	\caption{Example of an Young diagram and its transpose.}
\label{fig: gamma n}
\end{figure}
\fi

Representation theory provides a framework to study abstract algebraic structures by representing their elements as linear transformations of vector spaces. In particular, the representation theory of the symmetric group, the group of all permutations of a set, holds profound mathematical importance and has ties to diverse areas. Here, we briefly introduce only the relevant definitions needed to present our analysis. A comprehensive introduction to representation theory in the context of symmetric groups, non-commutative Fourier analysis, and random walks can be found in the books and monographs by Sagan \cite{sagan2013symmetric}, A. Terras \cite{terras1999fourier}, and Diaconis \cite{diaconis1988group}, respectively, as well as in the references therein. Much of the following material has been taken from those sources.

A \emph{representation} of a group \(G\) on a vector space \(V\) over a field \(F\) is a homomorphism \(\rho: G \to GL(V)\), where \(GL(V)\) is the group of invertible linear transformations of \(V\). Specifically, for each element \(g \in G\), there's an associated matrix \(\rho(g)\) that respects the group operation: \(\rho(g h) = \rho(g) \rho(h)\) for all $g, h \in G$. In our setting we take $F = \mathbb{C}$, the field of complex numbers. The \emph{dimension} of a representation $\rho$ corresponds to the dimension of its associated vector space $V$, denoted as $\dim \rho$. The representative matrices are $(\dim \rho) \times (\dim \rho)$ matrices, which can be made to be unitary.
A representation $\rho$ is termed \emph{irreducible} if no non-trivial invariant subspaces exist within it. This means the only subspaces of $V$ invariant under every transformation $(\rho(g), g \in G)$, are $V$ itself and the zero subspace. Henceforth we shall refer to irreducible representations as simply \emph{irreps} for brevity.
For a given representation \(\rho\) , the character of this representation, \(\chi_\rho\), is a function from \(G\) to the field $\mathbb{C}$ defined by the trace of the representation's matrix: \(\chi_\rho(g) = \text{Tr}(\rho(g))\). Following properties of $\chi_{\rho}$ will be useful: 
\begin{enumerate}
    \item $\chi_{\rho}(\mathbb{e}) = d_{\rho}$
    \item $\forall g,h \in G:\ \chi_{\rho}(gh) = \chi_{\rho}(hg)$ (cyclic property)
    \item $\forall g,h \in G:\ \chi_{\rho}(hgh^{-1}) = \chi_{\rho}(g)$ ($\chi_{\rho}$ is constant over the conjugacy classes)
\end{enumerate}
Elements \(g_1\) and \(g_2\) in \(G\) are termed \emph{conjugate} if there's an \(h \in G\) such that \(g_2 = h g_1 h^{-1}\). All elements conjugate to \(g_1\) form its \emph{conjugacy class}. In symmetric groups, conjugacy classes are characterized by a permutation's cycle type. A \emph{class function} on group \(G\) is a function \(f: G \to F\) (for some field $F$) that remains constant on conjugacy classes. Characters of representations are classic instances of class functions.

 A \emph{Young diagram} associated with a partition \(\lambda = (\lambda_1, \lambda_2, \ldots, \lambda_k)\)  of the number $n$ (that is $\lambda_1 + \cdots + \lambda_k = n$, denoted as $\lambda \vdash n$) consists of \(\lambda_1\) left-justified boxes in the top row, \(\lambda_2\) in the second, and so on. We will construct Young diagrams using the English convention, with row lengths decreasing or remaining constant from top to bottom.
 A \emph{Young tableau} fills this diagram with numbers from \(1\) to \(n\) such that entries in each row and column are increasing. A \emph{rim hook} is a set of boxes that can be removed from a Young diagram, leaving another Young diagram behind. 

The following text about \emph{Young normal form} is taken from \cite{gerhardt2003continuous} and has been modified to match the language of the present paper. A more comprehensive description can be found in \cite{wallace1984g,sagan2013symmetric}.
The symmetric group containing \( n \) elements provides a unique method to link partitions of \( n \) (which have a bijection to the conjugacy classes of $\sym$) to a full set of distinct, irreps of $\sym$. These distinct irreps are identified as the Young normal forms. 
A notable feature of these irreps is that each matrix entry within them is an integer. For each such representation, we may associate another irrep up to an isomorphism, with all its matrices being unitary. The irreducible, unitary representation corresponding to a specific partition \( \lambda \) is denoted as \( \rho_\lambda \), and its corresponding character is expressed as \( \chi_\lambda \).

Finally, we briefly discuss the \emph{Murnaghan-Nakayama Rule}, which is useful for computing the characters of certain irreps and conjugacy classes we used in this paper. This is a combinatorial method used to compute character values for symmetric group representations indexed by Young tableaux. The character of a permutation with cycle type \(\mu\) in the representation corresponding to a Young tableau of shape \(\lambda\) is determined by iteratively removing rim-hooks and summing the associated contributions. We need to define a few more terms before we can proceed. Given a Young diagram of shape (partition) $\lambda$ and a composition of $\mu = (\mu_1,\ldots,\mu_k)$ \footnote{A composition of $n$ is a partition where the order of the parts matters. For a given partition, the collection of compositions having the same parts corresponds to different permutations with the same cycle structure; hence, their characters are the same.}, a \emph{filling} of $\lambda$ using the \emph{content} from $\mu$ is a labeling of the cells of $\lambda$ such that $\mu_1$ cells are labeled with $1$, $\mu_2$ cells are labeled with $2$, and so on. Additionally, the labeling must satisfy the following two conditions:
(1) the cells corresponding to the same label have the shape of a rim-hook, and (2) labels are non-decreasing along rows and columns. Such a filling of the shape is called a \emph{rim-hook tableau}.
The \emph{leg-length} (denoted as $ll()$) of a rim-hook $\zeta$ is defined as $(\text{number of rows spanned by } \zeta) - 1$. The \emph{sign} of a rim-hook tableau $T$ is given by:
\begin{align*}
    \sign(T) = (-1)^{\sum_{\zeta \in T}{ll(\zeta)}}
\end{align*}
Then, the Murnaghan-Nakayama Rule provides a way to compute $\chi_{\lambda}(\mu)$:
\begin{align*}
    \chi_{\lambda}(\mu) = \sum_{T}\sign(T)
\end{align*}
where the sum is over all valid rim-hook tableaux for the pair $\lambda, \mu$.
\subsubsection{Fourier Transform Over Non-Commutative Groups}\label{sec: fouier trans}
The Fourier transform is a powerful tool in signal processing and applied mathematics, enabling the analysis of a signal's frequency content.
In the case of groups ($G$), the Fourier transform of a function $f: G \to \mathbb{C}$ performs a basis change from $\{\delta_g\mid g \in G\}$ to $\{\rho[i,j]\mid 1 \le i,j \le \dim \rho\}$. Here, $\rho[i,j]$ is the $(i,j)^{\text{th}}$ entry of the matrix presentation of $\rho$ for different group elements, thus constituting a function of the form $G \to \mathbb{C}$.
As observed, in the case of non-commutative groups, the Fourier transform takes a more complex form than in the commutative case, owing to the fact that the  irreps are themselves linear operators.
More formally, the Fourier transform over finite groups is defined as:
\begin{align*}
    \hat{f}(\rho) = \sum_{g \in G}f(g)\rho(g)
\end{align*}
In the case where $f$ is a class function, this simplifies to:
\begin{align*}
    \hat{f}(\rho) = \frac{1}{\dim \rho}\left(\sum_{[g]}f([g])\chi_\rho(g)|[g]|\right)I_{\dim \rho \times \dim \rho}
\end{align*}
where the sum is over all conjugacy classes in $G$, and $|[g]|$ denotes the size of $[g]$. We shall use the latter expression when computing certain projections with respect to the Szegedy walk operator.

\subsection{Quantizing Markov Chains: Szegedy Walk}\label{sec: szegedy walk}
Szegedy developed the framework of quantizing a Markov chain \cite{szegedy2004quantum}, which was then used to derive a quantum speedup of random-walk-based search algorithms on graphs. Here, we borrow Szegedy's terminology. Let $X$ and $Y$ be two parts of a bipartite graph, and let $P$ and $Q$ be probabilistic maps from $X$ to $Y$ and from $Y$ to $X$, respectively.
For any $x \in X$ and $y \in Y$ let
\begin{align*}
    \ket{\phi_x} = \sum_{y \in Y}\sqrt{P_{xy}}\ket{x,y} \\ \mbox{and} \\
    \ket{\psi_y} = \sum_{x \in X}\sqrt{Q_{yx}}\ket{x,y}
\end{align*}
Further, let $A$ (resp. $B$) be matrix composed of the column vectors  $\{\ket{\phi_x}\}$ (resp. $\{\ket{\psi_x}\}$). Define two reflection  operators - $R_A = 2AA^\dagger - I$ and $R_A = 2BB^\dagger - I$. Finally the quantum walk unitary is defined as 
\begin{align*}
    \bw = R_BR_A.
\end{align*}
The formulation above generalizes coined quantum walks on regular graphs in the following sense. To provide some intuition about this definition, we briefly introduce coined quantum walks. A classical random walk on vertices cannot be directly quantized into a unitary operator on the Hilbert space spanned by the vertices. To create a unitary, one has to lift the space on which the quantum walk takes place to a product of two Hilbert spaces: 1) a ``coin space," which is used to propagate amplitudes from a vertex to its neighbors in superposition, and 2) a shift or move operator that transfers the walker from the current vertex to its neighbors. More formally, 
the state of such a particle at any moment is described by a vector in the Hilbert space ${\cal H}$, with a basis set $\{\ket{c,x} \mid c \in C \text{ and } x \in X\}$ (standard basis), for some $|C|$-regular graph with vertex set $X$. Thus, we can express ${\cal H}$ as ${\cal H} = {\cal H}_X \otimes {\cal H}_C$.
The space ${\cal H}_X$ describes the position of the particle over the vertices.
${\cal H}_C$ is the coin space, which describes the state of the particle's internal degrees of freedom (sometimes referred to as the particle's \emph{chirality}).
One step of the walk consists of successively applying the two unitaries $U_C \otimes I_X$ and $\Lambda$, where $\Lambda = \sum_{c \in C, x \in X}\ketbra{c,c(x)}{x,c}$. Here, $c(x)$ denotes the $c^{th}$ neighbor of $x$. 
The shift operator $\Lambda$ moves the walker to its neighboring vertex in superposition. The coin operator $U_C$ determines how the amplitudes spread to neighboring vertices, acting like a quantum analogue of a classical $|C|$-sided die.
For graphs with arbitrary vertex degree, $U_C$ is replaced by the reflection operator $R_A$. It is easy to see that if we take $C = Y$, then $\Lambda R_A \Lambda = R_B$, where shift operator is generalized as $\Lambda = \sum_{x \in X, y \in 
Y}\ketbra{y,x}{x,y}$.
The spectral properties of the walk operator $\bw$ are closely related to the discriminant matrix $D$, whose entries are defined as $D_{xy} = \sqrt{P_{xy}Q_{yx}}$.
In the setting of quantized random walks, we shall take $X = Y = \mathcal{S}_n$. Furthermore, the transition probabilities are assumed to be uniform, and as such, $P$ is symmetric. Thus, $D = P$. Specifically,
\begin{align}\label{eqn: trans prob}
    D_{\pi\sigma} = \begin{cases}
    \frac{1}{d} \hspace{1cm} \mbox{if $\pi^{-1}\sigma \in S$}\\
    0 \hspace{1cm} \mbox{otherwise}
    \end{cases}
\end{align}

\ifx false
\subsubsection{Coin operators}
For $d \ge 3$ the Grover operator $D$ (reflection about the mean) is defined as follows. $D$ is also commonly known as the diffusion operator.  It is defined as: $D = 2\ketbra{\psi}{\psi}-I$, where $\ket{\psi} = \frac{1}{\sqrt{d}}\sum_{s \in S}\ket{s}$ is the uniform superposition over the basis states. $D$ acts only on the coin space ${\cal H}_S$. 
 Let $\delta_{ij}$ be the Kronecker delta function.
In the matrix notation $(i,j)^{th}$ entry of $D$ is given by:
 $D_{ij} = \delta_{ij}a + \left(1-\delta_{ij}\right)b$
 where $a = \frac{2}{d}-1$ and $b = \frac{2}{d}$.
When $|S| = 2$ we consider the Hadamard operator $H = \frac{1}{\sqrt{2}}\begin{bmatrix}1 & 1 \\ 1 & -1\end{bmatrix}$ or the operator $\frac{I+iX}{\sqrt{2}}$. Here $X$ is the  not gate. It has been shown that the propagation of the walk on the line when $C = \frac{I+iX}{\sqrt{2}}$ is symmetric \cite{lipton2014quantum} as opposed to $H$ which has a heavy tail on one side. 

\subsubsection{The $\Lambda$ operator}
The shift operator $\Lambda = \sum_{s \in S, g \in G}\ketbra{s,gs}{s,g}$. In literature it is sometimes referred to as the move operator to distinguish it from some of its extensions.
$\Lambda$ sends the walker with internal chiral state $\ket{s}$ and at position $g$ along the edge $s$ to $gs$.
In the matrix form, $\Lambda$ is a $dn \times dn$ block diagonal matrix with $d$ blocks. There is a block corresponding to each $s \in S$. The block corresponding to $s$ is the $n \times n$ permutation matrix associated with the action of $s$ on $G$.
A more general version of $\Lambda$ also permutes the basis in ${\cal H}_S$. Specifically, $\Lambda_{\pi} = \sum_{s \in S, g \in G}\ketbra{\pi(s),gs}{s,g}$. In the case of the grid graph, $\pi$ performing a directional flip ($\ket{\uparrow} $ to $\ket{\downarrow}$ and $\ket{\leftarrow} $ to $ \ket{\rightarrow}$ and vice versa), gives rise to the so-called flip-flop walk \cite{shenvi2003quantum}.

\fi
\subsubsection{Instantaneous and Limiting Distribution}
Given the initial state $\ket{\psi_0}$, the state after $t$ steps of the walk is represented as:
\begin{align*}
\ket{\psi_t} = \bw^t \ket{\psi_0}.
\end{align*}
As previously discussed, the basis of the Hilbert space in which the walk takes place consists of pairs of permutations from $\sym$, and we refer to this as the standard basis. From this point onward, we assume that all measurements are performed in this basis.
The probability of sampling a permutation $\pi$—more specifically, observing it in the first register—of $\sym$ after $t$ steps of the walk is given by:
\begin{align*}
    P_t[\pi \mid \psi_0] = \sum_{\sigma \in \sym}
    \norm{\bra{\pi,\sigma}\bw^t\ket{\psi_0}}_2^2
\end{align*}
Since $\bw$ is unitary, $\ket{\psi_t}$ exhibits periodicity \cite{aharonov2001quantum}, provided that $\ket{\psi_0}$ is not an eigenvector of $\bw$. Generally, $P_t$ does not converge. However, the time-averaged distribution, defined below, does converge as $T \to \infty$:
\begin{align*}
    \overline{P}_T[\pi \mid \psi_0] = \frac{1}{T}\sum_{t=0}^{T-1}P_t[\pi \mid \psi_0]
\end{align*}
$\overline{P}_T$ can be interpreted as the expected value of the distribution $P_t$ when $t$ is selected uniformly at random from the set $\{0,\ldots,T-1\}$.

In this paper, we are interested in upper-bounding $\overline{P}_T[\psi_{[n]} \mid \psi_0]$, where $\psi_{[n]}$ is the state representing the uniform superposition of pairs $(\pi,\sigma)$ in which the first register is an $n$-cycle. For the case of classical random walk this is analogous to determining the probability of sampling an $n$-cycle.
$\overline{P}_T[\ \mid \psi_0]$ defines a distribution $\overline{\mathcal{D}}_{n,\psi_0,S}$ on $\sym$, which depends on the initial state $\ket{\psi_0}$ and the choice of the generating set $S$. Since in our analysis both $S$ and $\ket{\psi_0}$  are fixed we simply use $\overline{\mathcal{D}}_n$ to denote this time averaged distribution.



\subsubsection{Continuous Time Quantum Walks and Average Mixing Matrix}
Here we briefly introduce some notion related to continuous time walks that will be useful to compare this work with some previous and recent results in the domain of continuous time quantum walks. A comprehensive introduction to concepts presented here can be found in \cite{coutinho2021graph} and the references therein. Let $A$ be the adjacency matrix of a undirected regular graph with vertex set $X$. $A$ is hermitian and as such $e^{itA}$ defines a Hamiltonian evolution on the Hilbert space spanned by the vertices of the graph. Let $U(t) = e^{itA}$. $U(t)$ is known as the continuous time quantum walk operator. It is important to note here that $U(t)$ acts directly on the vertices, which is not possible in the discrete setting.
In the setting of continuous time quantum walk a there is another notion of time average distribution - the \emph{average mixing matrix} \cite{godsil2013average}. 
Let $(A \circ B)_{x,y} = A_{x,y} B_{x,y}$ denote the \emph{Schur product} of $A$ and $B$. Then $M(t) = U(t)\circ U(-t)$, which is a doubly stochastic matrix. On the standard basis spanned by the vertices, the collection $\{\bra{x}M(t)\}$ gives rise family of probability densities - which can be interpreted as the resulting distribution after evolving for time $t$ starting from the vertex $x \in X$. To define the time average distribution , we can work with the time average version of $M(t)$, denoted as $\overline{M}(t)$, called the average mixing matrix.
:
\begin{align*}
    \overline{M}(t) = \frac{1}{T} \int_0^T M(t) dt. 
\end{align*}
Recently, average mixing matrix has been extended in the discrete setting \cite{sorci2022average}. In the language of this paper, we can define an average mixing matrix for the Szegedy walk operator as follows:
\begin{align*}
    \overline{M}_{xy} = \frac{1}{T}\sum_{t = 0}^{T-1}\sum_{\sigma \in S}P_t\left [\ket{x,x \sigma}\mid\ A\ket{y}\right]
\end{align*}
Here, the matrix $A$ is the matrix we defined earlier when introducing the Szegedy walk.
We can interpret $\overline{M}_{xy}$ as the average probability of observing $x$ in the first register, starting from the state $A\ket{y}$. In this context, the initial state is a superposition over the outgoing\footnote{Even though the walk takes place on an undirected graph, we may assign an orientation to an edge with respect to the vertex where the walker is situated, treating it as the tail of the edge. Since we are only considering bipartite walks, the walker can be present at most at one of the endpoints.
} edges from $y$, according to the amplitude distribution $\ket{\phi}_{y}$.
The average mixing matrix can be useful for studying the average limiting behavior of the quantized Markov chain. In particular, if the Markov chain $P$ mixes to a uniform distribution, then an analogous notion can be considered in the quantum case, where we are interested in how close $\overline{M}$ is to the matrix $\frac{1}{n}J$, with $J$ being the matrix whose all entries are 1. 
If $\overline{M}$ equals $\frac{1}{n}J$ in the limit, then we say the chain exhibits average uniform mixing.


\section{Previous and Related Work}
\label{sec: prev and releted}
\paragraph{Discrete Time Quantum Walk on Groups.}
In their seminal paper \cite{aharonov2001quantum}, Aharonov et al. presented several results on DTQWs. They characterized the convergence behavior of walks on abelian groups, showing that the time-averaged distribution converges to the uniform distribution whenever the eigenvalues of $U$ are all distinct. They also provided an $O\left(\frac{n\log n}{\epsilon^3}\right)$ upper bound on the mixing time for $\mathbb{Z}_n$ (the cycle graph), and proved some lower bounds in terms of the graph's conductance.
Following their introduction, DTQWs have been studied for several graph families. Nayak and Vishwanath \cite{nayak2000quantum} conducted a detailed analysis for the line using Fourier analysis, demonstrating that the Hadamard walk mixes almost uniformly in only $O(t)$ steps, achieving a quadratic speedup over its classical counterpart. Moore and Russell \cite{moore2002quantum} analyzed the Grover walk on the Cayley graph of $\mathbb{Z}^n_2$ (also known as the hypercube), showing an instantaneous mixing time of $O(n)$, which beats the classical $\Omega(n \log n)$ bound.
Acevedo and Gobron \cite{acevedo2005quantum} studied quantum walks for certain Cayley graphs, providing several results for graphs generated by free groups in particular. D’Ariano et al. \cite{d2016virtually} investigated the case where the group is virtually abelian, a condition that allowed them to reduce the problem to an equivalent one on an abelian group with a larger chiral space dimension, employing the Fourier method introduced in \cite{nayak2000quantum}.
More recently, DTQWs on the Dihedral group $D_n$ have been studied by Dai et al. \cite{dai2018discrete} and Sarkar and Adhikari \cite{sarkar2023discrete}. Since $D_n$ is isomorphic to the semi-direct product $\mathbb{Z}_n \rtimes \mathbb{Z}_2$, the Fourier approach introduced in \cite{nayak2000quantum} is applicable once again. Using this method, the authors in \cite{dai2018discrete} provided a spectral decomposition of $U$ for the Grover walk. In \cite{sarkar2023discrete}, the authors study the periodicity and localization properties of the walk using generalized Grover coins.
A detailed survey of various types of quantum walks, including DTQW, can be found in \cite{venegas2012quantum}, with additional references therein. A survey specifically addressing DTQWs on Cayley graphs is available in \cite{knittelquantum}.

\paragraph{Quantum Walk on the Symmetric Group.}
In a previous work by the author~\cite{banerjee2022discrete}, DTQW on the symmetric group was studied using the coin-based model. Utilizing the Fourier transform (see Section~\ref{sec: fouier trans}), a recurrence relation was derived for the amplitudes of $\ket{\phi_t}$, from which a ``sum-over-paths" type expression was determined for the amplitudes. It was also determined under which conditions the amplitudes are class functions. Prior to this, as indicated earlier, Gerhardt and Watrous~\cite{gerhardt2003continuous} studied the continuous time quantum walk model on the symmetric group. They showed that when $S$ is the set of transpositions, the time-averaged distribution is far from the uniform distribution. They explicitly calculated the probability of reaching an $n$-cycle starting from $\mathbb{e}$ by expressing the eigenstates of $\bw$ using the characters of ${\cal S}_n$. They also considered the generating set (for $\mathcal{A}_n$) consisting of all $p$-cycles (where $p$ is odd) and derived similar results as in the transposition case.

\paragraph{Average limiting behavior.}

Considerable research has been conducted on studying the limiting behavior of quantized Markov chains, as mentioned earlier. More recently, the properties of the average mixing matrix have been explored, especially for continuous-time chains.
For example, in \cite{sorci2022average,chan2023pretty}, the entries of the average mixing matrix in the limit as $t \to \infty$ were expressed using the projectors in the spectral decomposition of the walk operator. One of the interesting questions with respect to limiting behavior is whether a quantized Markov chain exhibits average uniform mixing. To this end, the authors in \cite{sorci2022average} have constructed a family of Markov chains whose quantized versions do exhibit such average mixing behavior, and have shown that average uniform mixing of the continuous-time quantum walk implies the same for its discretized version.


\section{Eigendecomposition of $\bw$ over the irreps of $\mathcal{S}_n$}
Gerhardt and Watrous used representation theory to express the eigenstates using the irreps of ${\cal S}_n$. This method is effective due to the fact that the walker's Hamiltonian is completely specified by the adjacency matrix of $\Gamma_n$. In the discrete case, as $\bw$ acts on a larger space, this decomposition becomes more complex for an arbitrary generating set. However, there is at least one special case where we can directly apply their Fourier method.

This special case occurs when the generating set $S$ forms a group itself. For the amplitudes to be uniform over the conjugacy class, which is a necessity for using Fourier analysis over the characters of the irreps., the generating set $S$ must be conjugate invariant. However, the only non-trivial subgroup of ${\cal S}_n$ that is also conjugate invariant is the \emph{alternating group} ${\cal A}_n$, which is the subgroup of all even permutations in ${\cal S}_n$. In this scenario, it becomes possible to factorize the space of irreps. for the walk on $\Gamma = ({\cal S}_n, {\cal A}_n)$, and determine the spectral decomposition of $\bw$ using the characters of both ${\cal S}_n$ and ${\cal A}_n$. To  overcome this issue we use the Szegedy walk formalism, which considers an even larger coin-space, as introduced earlier.


As Szegedy showed, the dynamics of the walk operator $\bw$ can be determined from the discriminant matrix $D$. 
Since $D$ is Hermitian (in fact, symmetric), the singular values of $D$ lie in the interval $[0,1]$. We index the singular values $\lambda_{\mu}$ of $D$ using the conjugacy classes $\mu$ of $G$, which are the partitions of $n$.
For each $\lambda_\mu$, if the corresponding eigenvalue is also $\lambda_\mu$, then the left and right singular vectors are equal (and are equal to the corresponding eigenvector); otherwise, they differ by a minus sign. For the former case, we use $\ket{\lambda_\mu}$ to denote both the left and the right singular vectors. For the latter case, without loss of generality, we use $-\ket{\lambda_\mu}$ to denote the left singular vector by appropriately choosing the sign of the corresponding eigenvector.
Let $\Pi_{\col(A)}$ (resp., $\Pi_{\col(B)}$) denote the projector onto the column space of $A$ (resp., $B$), and let $\Pi_{\ker(A)}$ (resp., $\Pi_{\ker(B)}$) denote the projector onto the orthogonal complement of the column space of $A$ (resp., $B$).
We can restate the \emph{spectral lemma} from \cite{szegedy2004quantum} in the language of this paper, which will be used in our subsequent analysis.

\begin{lemma}[modified Lemma-1 from \cite{szegedy2004quantum}]\label{lem: spectra}
    Let $\lambda_1,\ldots \lambda_l$ (with multiplicity) are the sequence of singular values of $D$ in the interval $(0,1)$ and $\Tilde{\lambda_\mu}$ be the eigenvalue corresponding to $\lambda_\mu$. Then the eigenvalues and eigenvectors of the walk operator $\bw$ is  $e^{\pm 2i\cos^{-1}\lambda_1}\ldots e^{\pm 2i\cos^{-1}\lambda_l}$ and $(A-(\sign \Tilde{\lambda_1})e^{\pm i\cos^{-1}\lambda_1}B)\ket{\lambda_1} \ldots (A-(\sign \Tilde{\lambda_l})e^{\pm i\cos^{-1}\lambda_l}B)\ket{\lambda_l}$ respectively (up to a normalization). Additionally, corresponding to the singular value  $1$, $\bw$ acts as the identity ($I$) on the space $\col(A)\cap \col(B) \oplus \ker(A)\cap \ker(B)$ and corresponding to the singular value $0$, $\bw$ acts as $-I$ on the space  $\col(A)\cap\ker(B) \oplus \ker(A)\cap\col(B)$.
\end{lemma}
From the definition of $D$, as presented in \cite{gerhardt2003continuous}, we may define a class function $f$ on $G$ such that $D_{\pi\sigma} = f(\pi^{-1}\sigma)$, provided that $S$ is conjugate invariant (which holds true in our case). We will employ the spectral decomposition of $D$ in terms of the characters of $\mathcal{S}_n$, as given in \cite{gerhardt2003continuous}, which takes advantage of the fact that the entries of $D$ behave as a class function. This is a special case of a more general result by Diaconis \cite{diaconis1988group} presented earlier. It has also been shown that the basis of the irreps are the eigenvectors of $D$. This can, in turn, be utilized to derive the spectra of the walker's Hamiltonian by exponentiating the corresponding eigenvalues of $D$. In the discrete case, the relationship between $D$ and the walk operator $\bw$ is somewhat more subtle, and the remainder of this section is devoted to elucidating it.

Recall that the conjugacy classes in $G$ have a one-to-one correspondence with the collection of non-isomorphic irreps of $G$. These irreps can be expressed using the so-called \emph{Young normal form}, whose matrix entries are all integers.
Let $\rho_{\mu}$ denote the irrep corresponding to the conjugacy class $\mu$ (with size given by $\abs{\mu}$), which is a partition of $n$ ($\mu \vdash n$). Define $\rho_{\mu,i,j}(g) = \rho_\mu(g)[i,j]$. The vectors $\ket{\rho_{\mu,i,j}}$ (in the $G$-module $\mathbb{C}[G]$ over the field of complex numbers) form an orthonormal basis corresponding to the irrep $\rho_\mu$. 
For reference, we restate Lemma 6 from \cite{gerhardt2003continuous}, which gives the expressions for the eigenvalues of $D$ in terms of the characters of $G$.

\begin{lemma}[modified lemma-6 from \cite{gerhardt2003continuous}]\label{lem: eigen char}
    Given $D, f$ as above, then $\ket{\rho_{\mu,i,j}}$ are the eigenvectors of $D$ with the eigenvalues,
    \begin{align}\label{eq: lemm 2 eigenvalue}
       \Tilde{\lambda_{\mu}} =  \frac{1}{\dim \rho_\mu} \sum_{\sigma \vdash n}|[\sigma]|f(\sigma)\chi_\mu(\sigma)
    \end{align}
\end{lemma}

\noindent Recall that $\lambda_{\mu} = \abs{\Tilde{\lambda_{\mu}}}$ are the singular values of $D$. Let $\kappa_\mu = e^{2i\cos^{-1}\lambda_\mu}$. 
Associated with each non-extremal singular value $\lambda_{\mu}$ ($\not \in \{0,1\}$) of $D$, there is a collection of eigenvectors $\{\ket{\rho_{\mu,i,j}}\}$, where $1 \le i,j \le \dim \rho_\mu$. The $g^{th}$ component of $\ket{\rho_{\mu,i,j}}$ is given by $\rho_\mu(g)[i,j]$. 
The projectors onto the space spanned by $\ket{\rho_{\mu,i,j}}$ are given by 
\begin{align*}
    \Pi_{\mu,i,j} = \frac{(\dim \rho_{\mu})\ket{\rho_{\mu,i,j}}\bra{\rho_{\mu,i,j}}}{n!},
\end{align*}
where $\frac{\dim \rho_\mu}{n!}$ is a normalization factor, since we have $\norm{\braket{\rho_{\mu,i,j}}{\rho_{\mu,i,j}}} = \frac{n!}{\dim \rho_{\mu}}$. Let $\Pi_\mu = \sum_{i,j} \Pi_{\mu,i,j}$ be the projector onto the column space of $\Tilde{\lambda_\mu}$.
Let the projectors corresponding to the subspaces $(\col(A)\cap \col(B)) \oplus (\ker(A)\cap \ker(B))$ and $(\col(A)\cap\ker(B)) \oplus (\ker(A)\cap\col(B))$ be $\Pi_{+1}$ and $\Pi_{-1}$, respectively.
Further, let $\bw_{A,B} = \Pi_{+1} - \Pi_{-1}$.
We can completely determine the spectral decomposition of $\bw$ using that of $D$, as given by the following lemma. Let $s_\mu = \sign \Tilde{\lambda_\mu}$, and let $\theta_\mu = \cos^{-1}\lambda_\mu$.

\begin{lemma}\label{lm: spectra w}
    Spectral decomposition of $\bw$ is given by,
    \begin{align*}
           \bw &= \sum_{\mu \vdash n} \frac{1}{2\sin^2\theta_\mu}\kappa_\mu (A-s_\mu\sqrt\kappa_\mu B)
           \Pi_\mu(A^\dagger-s_\mu\sqrt\kappa_\mu^* B^\dagger) +\\ & \sum_{\mu \vdash n}\frac{1}{2\sin^2\theta_\mu}\kappa_\mu^*(A-s_\mu\sqrt\kappa_\mu^* B)\Pi_\mu(A^\dagger-s_\mu\sqrt\kappa_\mu B^\dagger) + \bw_{A,B}
    \end{align*}

\end{lemma}
\begin{proof}
The proof directly follows from the preceding discussions and Lemma \ref{lem: spectra}. The factor $\frac{1}{2\sin^2 \theta_\mu}$ comes from normalizing the eigenvectors.
\end{proof}
The decomposition of $\bw$ can then be divided into two parts: one corresponding to the non-trivial eigenvalues, and the other corresponding to the two trivial ones, $\{-1,1\}$.
Let $\bw_\mu$ be the term corresponding to $\mu$.
Then,
\begin{align}\label{eq: decom of w}
    \bw^t = \sum_{\mu \vdash n}\bw_\mu^t + \bw_{A,B}^t
\end{align}
where, by a slight abuse of notation, we assume the sum above is over all partitions except those corresponding to the trivial eigenvalues.
It follows that $\bw_{A,B}^t = \Pi_{+1} + (-1)^t \Pi_{-1}$, since the projectors are mutually orthogonal.
Expanding $\bw_\mu^t$ we get,
\begin{align}\label{eq: mu spectra}
  \nonumber  \bw_\mu^t = &  \frac{1}{2\sin^2 \theta_\mu}\left(\vphantom{\frac12} (\kappa_\mu^t+\kappa_\mu^{*t})A\Pi_{\mu}A^\dagger - s_\mu(\kappa_\mu^t\sqrt\kappa_\mu^*+\kappa_\mu^{*t}\sqrt\kappa_\mu)A\Pi_{\mu}B^\dagger\right. \\ \nonumber & \left. -   s_\mu(\kappa_\mu^t\sqrt\kappa_\mu+\kappa_\mu^{*t}\sqrt\kappa_\mu^*)B\Pi_{\mu}A^\dagger + (\kappa_\mu^{t+1/2}\sqrt\kappa_\mu^*+\kappa_\mu^{*t+1/2}\sqrt\kappa_\mu)B\Pi_{\mu}B^\dagger )\vphantom{\frac12}\right)\\ \nonumber
    = & \frac{1}{2\sin^2 \theta_\mu}\left(\vphantom{\frac12}\cos{2\theta_\mu t}(A\Pi_{\mu}A^\dagger+B\Pi_{\mu}B^\dagger)- s_\mu\cos{2\theta_\mu(t-1/2)}A\Pi_{\mu}B^\dagger \right. \\ & \left. - s_\mu\cos{2\theta_\mu(t+1/2)}B\Pi_{\mu}A^\dagger \vphantom{\frac12}\right)
\end{align}
From the above, we see that for any pair of initial and final states $\ket{\psi_0}$ and $\ket{\psi_t}$, respectively, an upper bound on the inner product $\bra{\psi_t}\bw\ket{\psi_0}$ depends on the projectors respective to the irreps. In particular, if the overlap is sufficiently low, then we can ignore the effect of the cosine terms (replacing them with 1 or -1, as appropriate) and still obtain a non-trivial upper bound. In the following, we use this approach to compute $\norm{\bra{\phi_{[n]}}\bw^t\ket{\phi_{\mathbb{e}}}}$.

\ifx false
We say a state $\ket{\psi} = \sum_{g,h}\alpha_{g,h}\ket{g,h}$ is a class function if $\alpha_{g,h} = \alpha'(g^{-1}h)$ for some class  $\alpha'$  on $\mathcal{S}_n$.
\begin{corollary}
    If $\ket{\psi}$ is a class function then so is $\bw \ket{\psi}$. 
\end{corollary}
\begin{proof}
    From the spectral decomposition of $\bw$ we only need to show that the spectral idempotents of $\bw$ preserve the class function property of $\ket{\psi}$, since if a collection of states are class functions, then so are there linear combinations.
\end{proof}
\fi 
\section{Divergence of $\bw$ from uniform mixing}
In \cite{gerhardt2003continuous}, the divergence of the instantaneous distribution from the uniform distribution was confirmed by upper bounding the probability of being at an $n$-cycle of $\mathcal{S}_n$. This probability was shown to be exponentially smaller than in the classical case. Specifically, they showed that the probability of being at some $n$-cycle is $O(2^{-2n}/(n+1)!)$, as compared to $\frac{1}{n}$ in the classical case (uniform distribution). In the discrete case that we study here, the walk takes place on the edges of the graph. We show that the instantaneous probability of observing an $n$-cycle in the first register is upper bounded away from $\frac{1}{n}$ by a function which is $\tilde{O}\left(\frac{c^n}{n!}\right)$. Here, $\tilde{O}$ represents the soft-O notation, but in our case, we ignore any functions up to polynomial in $n$ (which are most likely due to artifacts from our techniques). Since the bound we provide is on the instantaneous probability, it also implies that the average mixing probabilities (entries of the average mixing matrix) are exponentially far from uniform.

In the following, we use ``$g, h$'' to denote generic permutations from $\sym$, aiming to avoid confusion arising from our slight notational abuse. Specifically, we use $\mu$ to denote a generic conjugacy class that contains the permutation $\mu$.
Recall,
\begin{align*}
    \ket{\psi_{[n]}} &= \frac{1}{\sqrt{(n-1)!d}}\sum_{g \in [(n)], s \in  S}\ket{g,gs} \\
    \ket{\psi_{\mathbb{e}}} &= \frac{1}{\sqrt{d}}\sum_{s \in  S}\ket{\mathbb{e},s}
\end{align*}
Also, $d = |S| = {n \choose 2}$ is the degree of $\Gamma^{\musDoubleFlat}_n$. We will provide an upper bound for 
$\norm{\bra{\phi_{[n]}}\bw^t\ket{\phi_{\mathbb{e}}}}$.
The state $\phi_{[n]}$ is the uniform superposition over all outgoing edges of $n$-cycles, analogous to the state in classical walks and continuous time quantum walks, which is the uniform superposition of all $n$-cycles.
In the discrete case, another possibility is to compute $\norm{\bra{g,gs}\bw^t\ket{\phi_{\mathbb{e}}}}$ for some arbitrary $n$-cycle $g$. By then summing over $S$, we find that $\sum_{s}\norm{\bra{g,gs}\bw^t\ket{\phi_{\mathbb{e}}}}$ gives the instantaneous probability of observing $g$ in the first register after $t$ steps.
In both cases, we start from the state $\ket{\phi_{\mathbb{e}}}$, which is the uniform superposition over all outgoing edges from the identity permutation. This choice of initial state is consistent with our definition of the average mixing matrix and is symmetric with respect to the generating set $S$. Unfortunately, to compute the aforementioned quantity, we need to have direct access to the entries of the irreps, as the final state does not span all the basis vectors of a conjugacy class. Thus, we make do with computing the former expression. We discuss this issue briefly at the end of this section.

This section is organized as follows. First, in Section \ref{sec: 5 -prelim}, we derive expressions for the $\Tilde{\lambda}_\mu$'s using the character theory of $\sym$. In Section \ref{sec: divergence phi n}, we prove Theorem~\ref{thm: main theorem}. Finally, we discuss the issue related to computing $\sum_{s}\norm{\bra{g,gs}\bw^t\ket{\phi_{\mathbb{e}}}}$ in Section~\ref{sec: instan prob}.






\subsection{Computing $\Tilde{\lambda_\mu}$ }\label{sec: 5 -prelim}
Here, we derive expressions for $\Tilde{\lambda}_\mu$, which will be used later for bounding the probabilities.
Recall that our walk takes place on $\Gamma^{\musDoubleFlat}_n$, derived from $\Gamma_n = (\sym, S)$, where $S$ is the class of all transpositions. Furthermore, the entries of the discriminant matrix $D_{gh} = f(g^{-1}h)$ form a class function over $\sym$ (equation~\ref{eqn: trans prob}). 
Specifically we define:
\begin{align*}
    f(g) = \begin{cases}
    \mbox{$\frac{1}{d}$ if $g \in S$}\\
    \mbox{$0$ otherwise}
    \end{cases}
\end{align*}
In the expression for characters, we will use $[\sigma]$ instead of $S$ to denote the class of transpositions going forward.
The above definition of $f$ simplifies the expression for the eigenvalues $\tilde{\lambda}_{\mu}$ given in equation~\ref{eq: lemm 2 eigenvalue} as follows:
\begin{align}\label{eq: lambda}
\tilde{\lambda}_{\mu} = \frac{\chi_\mu([\sigma])}{\dim \rho_\mu}
\end{align}
Next, we set out to compute the values of $\chi_\mu([\sigma])$ and the dimension of $\rho_\mu$.
Let $\mu = (\mu_1,\ldots,\mu_l)$, where $\mu_1 + \cdots + \mu_l = n$. It is known that \cite{re1950some}: 
\begin{align}\label{eqn: char mu sigma}
    \chi_\mu([\sigma]) = \frac{\dim \rho_\mu}{n(n-1)}\left(\sum_{j=1}^{l} (\mu_j-j+1)(\mu_j-j) - j(j-1) \right)
\end{align}
Fortunately, we only need to compute these quantities for a subset of partitions ($\mu$'s) of $n$, which will greatly simplify our analysis.
Specifically, we assume $\mu \in \Xi_n$, where $\Xi_n$ is the collection of partitions of $n$ having the following property: There exist four non-negative integers $\mu_1 \ge \mu_2 \ge 1$, and $r , l \ge 2$ such that $\mu = (\mu_1, \mu_2, 2^{r-2},1^{l-r})$. Here, we allow $\mu_2 = 1$, but in that case, $\mu$ can simply be written as $\mu = (k,1^{n-k})$ (that is, $\mu_1=k, \mu_2=1, r = 2, l = n-k+1$), and we say $\mu \in \Xi_{n,k}$.
The following fact is a simple application of the Murnaghan-Nakayama rule (see for example Chapter 4 in \cite{sagan2013symmetric}).
\begin{fact}\label{fc: 1}
Let $[n]$ set of all $n$-cycles of $\sym$.  Then 
\begin{align*}
    \chi_\mu([n]) = \begin{cases}
    \mbox{$(-1)^{n-k}$ \qquad if  $\mu \in \Xi_{n,k}$ and}\\
    \mbox{$0$ \qquad \qquad \qquad otherwise}
    \end{cases}
\end{align*}
\end{fact}
\noindent  Let $[\tau_l]$ (where $1 \le l \le \lfloor n/2 \rfloor$) be the conjugacy class of all permutations with one cycle of length $l$ and another of length $n-l$. Then, the following two facts are immediately apparent:
\begin{fact}\label{fc: 2}
If $\sigma \in [\sigma]$ and  $\tau \sigma \in [n]$ then $\tau \in [\tau_l]$.
\end{fact}
\begin{fact}\label{fc: 3}
If  $\tau \in [\tau_l]$ for some $l$ and $\chi_\mu([\tau_l]) \ne 0$ then $\mu \in \Xi_n$.
\end{fact}
\begin{fact}\label{fc: xi size}
Size of the set $\Xi_n$ is $O(n^3)$.
\end{fact}
\noindent Again using the Murnaghan-Nakayama rule we get:
\begin{fact}\label{fc: 4}
\begin{align*}
    \chi_{\mu}([\tau_l]) = \begin{cases}
    \mbox{$(-1)^{n-k-1}$ \qquad if  $\mu \in \Xi_{n,k}$ and $l \ge k$  }\\
    \mbox{$0$ \qquad \qquad \qquad otherwise}
    \end{cases}
\end{align*}
\end{fact}

\noindent Although we do not need to know $\dim \rho_\mu$ in order to compute $\Tilde{\lambda}_\mu$ we will need this for our analysis in later sections.

\ifx false
Substituting the above in Equation \ref{eq: lambda} we get:
\begin{align*}
    \Tilde{\lambda}_{(k,1,\ldots,1)} = \frac{1}{n(n-1)}\sum_{j=1}^{n-k+1} (\mu_j-j+1)(\mu_j-j) - j(j-1)
\end{align*}

$\Tilde{\lambda_{\mu}}$ and setting $\mu = (k,1,\ldots,1)$ we get,
\begin{align*}
    \Tilde{\lambda}_{(k,1,\ldots,1)} = \frac{1}{n(n-1)}\sum_{j=1}^{n-k+1} (\mu_j-j+1)(\mu_j-j) - j(j-1) = -\frac{n-2k+1}{n-1}
\end{align*}

We see that $s_{(k,1,\ldots,1)} = 1$ iff $2k > n$.
Further we have,
\begin{align*}
    \chi_\mu([n]) = \begin{cases}
    \mbox{$(-1)^{n-k}$ if $\mu = (k,1,\ldots,1)$ and}\\
    \mbox{$0$ otherwise}
    \end{cases}
\end{align*}
and 
\begin{align*}
    \chi_{(k,1,\ldots,1)}(\tau_l) = \begin{cases}
    \mbox{$(-1)^{n-k-1}$ if $\mu = (k,1,\ldots,1)$ and $l \ge k$  }\\
    \mbox{$0$ otherwise}
    \end{cases}
\end{align*}

where $\tau_l$ is a permutation consisting of a $l$-cycle and a $n-l$ cycle. Also note that $\dim \rho_{(k,1,\ldots,1)} = {n- 1 \choose k-1}$, which can be easily seen by using the Murnaghan-Nakayama rule.
\fi
\begin{lemma}\label{lm: dimension mu}
Let $\mu \in \Xi_n$ then
\begin{align*}
    \dim {\rho_\mu} = \begin{cases}
        {n-1 \choose k-1}\qquad\qquad\qquad\qquad\qquad\qquad\qquad\qquad\qquad\qquad\qquad \mbox{if $\mu \in \Xi_{n,k}$}\\
        \frac{n!(\mu_1-\mu_2+1)(l-r+1)((\mu_1+l-1)(\mu_1+r-2)(\mu_2+l-2)(\mu_2+r-3))^{-1}}{(\mu_1-1)!(\mu_2-2)!(l-1)!(r-2)!} \quad\quad \mbox{otherwise}
    \end{cases}
\end{align*}
\end{lemma}
\begin{proof}
The proof of the lemma follows directly from the application of the \emph{hook-length} formula, given by:
\begin{align*}
    \dim {\rho_\mu} = \frac{n!}{\prod_{i,j}h_\mu(i,j)}
\end{align*}
where $h_\mu(i,j)$ represents the hook-length of the cell $(i,j)$ in the Young diagram of the partition $\mu$. The case $\mu \in \Xi_{n,k}$ is straightforward, and we will focus on deriving the latter case. 
Since $\mu \in  \Xi_n \setminus \Xi_{n,k}$, we can express $\mu$ as $(\mu_1,\mu_2,2^{r-2},1^{l-r})$, where $\mu_1 \ge \mu_2 \ge 2$. Below, we provide the explicit hook lengths for a cell $(i,j)$, from which the lemma immediately follows.
    \begin{align*}
        h_\mu(i,j) = \begin{cases}
            (\mu_1+l-1)\qquad \mbox{if $i=j=1$}\\
            (\mu_1+r-2)\qquad \mbox{if $i=1, j= 2$}\\
            (\mu_1 - j + 2)\qquad \mbox{if $i=1, 2 < j \le \mu_2$}\\
            (\mu_1 - j + 1)\qquad \mbox{if $i=1, \mu_2 < j \le \mu_1$}\\
            (\mu_2 + l - 2)\qquad \mbox{if $i=2, j = 1$}\\
            (\mu_2 + r - 3)\qquad \mbox{if $i=2, j = 2$}\\
            (\mu_2 - j + 1)\qquad \mbox{if $i=2, 2 < j \le \mu_2$}\\
            (l - i + 2)\qquad \mbox{if $2 < i \le r, j = 1$}\\
            (r - i + 1)\qquad \mbox{if $2 < i \le r, j = 2$}\\
            (l - i + 1)\qquad \mbox{if $r < i \le l, j = 1$}
        \end{cases}
    \end{align*}
\end{proof}
\begin{lemma}\label{lm: mu tau l}
If $\mu \in \Xi_n$ and $\tau \in [\tau_l]$ then $\chi_\mu([\tau_l]) \in \{-1,0,1\}$.
\end{lemma}
\begin{proof}
Given $\tau \in [\tau_l]$, it is necessary for the partition $\mu$ to be in $\Xi_n$ in order for $\chi_\mu([\tau])$ to be non-zero, as any valid rim-hook tableaux filled using the labels from $\{1,2\}$ (according to the composition $\tau$) must take the form of partitions in $\Xi_n$.
As it turns out, for any such shape $\mu$, there is at most one way to create a rim-hook tableau $T_{\mu,\tau}$. In other words, $\chi_\mu([\tau_k]) = \sign(T_{\mu,\tau})$ if $T_{\mu,\tau}$ exists, and $0$ otherwise.

\end{proof}
\noindent Next, we provide an upper bound on $\chi_\mu([\sigma])$, which is divided into Lemmas~\ref{lm: chi mu trivial} and~\ref{lm: bound chi mu sigma}.
\begin{lemma}\label{lm: chi mu trivial}
If $\mu \in \Xi_{n,k}$ then $\chi_{\mu}([\sigma]) = -\frac{n-2k+1}{n-1}{n-1 \choose k-1}$
\end{lemma}
\begin{proof}
The stated bound follows immediately from the expression for $\chi_\mu([\sigma])$ given in Equation \ref{eqn: char mu sigma}, combined with Lemma \ref{lm: dimension mu}, when we substitute the components of $\mu = (k,1^{n-k})$. Specifically,
    \begin{align*}
        \chi_{\mu}([\sigma]) = \frac{\dim \rho_\mu}{n(n-1)}\sum_{j=1}^{n-k+1} (\mu_j-j+1)(\mu_j-j) - j(j-1) = -\frac{n-2k+1}{n-1}{n-1 \choose k-1}
    \end{align*}        
\end{proof}

\begin{lemma}\label{lm: bound chi mu sigma}
    If $\mu \in \Xi_n \setminus \Xi_{n,k}$  then for any constant $\beta \ge \frac{81}{16}$ we have $ \abs{\chi_{\mu}([\sigma])} =  O\left (n^{6.5} \beta^{n}\right)$
\end{lemma}

\begin{proof}
    Substituting $\mu = (\mu_1,\mu_2,2^{r-2},1^{l-r})$ in Equation \ref{eqn: char mu sigma} we get:
    \begin{align*}
        \frac{n(n-1)}{\dim \rho_\mu}\chi_\mu([\sigma]) &= \sum_{j=1}^{l} (\mu_j-j+1)(\mu_j-j) - j(j-1)  =  \mu_1(\mu_1-1)+ (\mu_2-1)(\mu_2-2) - 2 \\ & + \sum_{j=3}^r(3-j)(2-j) - \sum_{j=3}^l j(j-1)  +  \sum_{j=r+1}^l(2-j)(1-j)\\
        & = \mu_1^2 + \mu_2^2 - \mu_1 - 3\mu_2 + l - l^2 - (r-3) r 
    \end{align*}
Hence,
\begin{align}\label{eq: char fraction}
    \chi_{\mu}([\sigma]) &= \frac{(\mu_1^2 + \mu_2^2 - \mu_1 - 3\mu_2 + l - l^2 - (r-3) r)(\mu_1-\mu_2+1)(l-r+1)}{n(n-1)(\mu_1+l-1)(\mu_1+r-2)(\mu_2+l-2)(\mu_2+r-3)}
    \times T_\mu
\end{align}
where $T_\mu = n!((\mu_1-1)!(\mu_2-2)!(l-1)!(r-2)!)^{-1}$.
Suppose that $\mu_i = \beta_i n$ for $i \in \{1,2\}$, and let $l = \beta_3 n$, $r = \beta_4 n$, with $\beta_1 \geq \beta_2$ and $\beta_3 \geq \beta_4$, where $\beta_i \geq 0$. Note that $\sum_i \beta_i = 1 + \frac{4}{n}$. Here, we ignore the fact that the terms $\beta_i n$ may not be integers, as it does not affect our analysis. Finally, by substituting Stirling's approximation formula for factorials ($n! = \Theta(n^{n+0.5}e^{-n})$), we obtain the claimed bound, as shown below.

\begin{align*}
    T_\mu & \lesssim \frac{n^{0.5+n}e^{-n+\left(n\sum_i{\beta_i}\right)}}{\prod_{i}(\beta'_i n)^{\left(\beta'_i n - 1 + (-1)^{i}\frac{1}{2}\right)}(\beta_i'n)^{\frac{3+(-1)^i}{2n}}} \\
    & \lesssim \frac{n^{0.5+n}}{\prod_{i}(\beta'_i n)^{\left(\beta'_i n - 1 + (-1)^{i}\frac{1}{2}\right)}} \\
    & \lesssim  \frac{n^{0.5+n - n\sum_{i}\beta'_i  + 4}}{\left(\prod_{i}{\beta_i^{'\beta'_i n}}\right)} = \frac{n^{6.5}}{\left(\prod_{i}{\beta_i^{'\beta'_i}}\right)^n} \lesssim n^{6.5}\beta^n 
\end{align*}
In the above derivation, $\lesssim$ denotes that the quantity on the left-hand side is less than or approximately equal to the quantity on the right-hand side up to a multiplicative constant. We set $\beta_i' = \beta_i - \frac{3+(-1)^i}{2n}$. To justify the last inequality, consider the fact that $0 < \beta'_i < 1$, which implies $\beta_i^{'\beta'_i} \ge 2/3$. Hence, taking $\beta \ge \frac{81}{16}$ completes the derivation.
Now, for the fraction on the right-hand side of Equation \ref{eq: char fraction}, we observe that the absolute value of the denominator is
\begin{align*}
   \abs{n(n-1)(\mu_1+l-1)(\mu_1+r-2)(\mu_2+l-2)(\mu_2+r-3)} \gtrsim n^4, 
\end{align*}
and the numerator is
\begin{align*}
   \abs{(\mu_1^2 + \mu_2^2 - \mu_1 - 3\mu_2 + l - l^2 - (r-3)r)(\mu_1-\mu_2+1)(l-r+1)} \lesssim n^4.
\end{align*}
Hence, we conclude that $\abs{\chi_{\mu}([\sigma])} \lesssim T_\mu$, which proves the lemma.
\end{proof}

\begin{lemma}\label{lm: lambda mu bound}
If $\mu \in \Xi_n$ and $\chi_\mu([\sigma]) \ne 0$, then for any non-trivial singular value, we have $\lambda_{\mu} \le 1-\frac{2}{n-1}$.
\end{lemma}
\begin{proof}
Combining Equations~\ref{eq: lambda} and \ref{eqn: char mu sigma}, we obtain
\begin{align}\label{eq: mu-j}
\nonumber \Tilde{\lambda_{\mu}} &= \frac{\chi_\mu([\sigma])}{\dim \rho_\mu} = \frac{1}{n(n-1)}\left(\sum_{j=1}^{l} (\mu_j-j+1)(\mu_j-j) - j(j-1) \right) \\  
&= \frac{\mu_1^2 + \mu_2^2 - \mu_1 - 3\mu_2 + l - l^2 - (r-3) r }{n(n-1)},
\end{align}
where the last equality follows from the proof of Lemma~\ref{lm: bound chi mu sigma}. 
If $\mu \in \Xi_{n,k}$, then from Lemma~\ref{lm: chi mu trivial}, we see that $\abs{\Tilde{\lambda_{\mu}}} \le \frac{n-3}{n-1}$. For any $\mu' \in \Xi_n \setminus \Xi_{n,k}$, it can be easily verified from Equation~\ref{eq: mu-j} that $\abs{\Tilde{\lambda_{\mu'}}} \le \abs{\Tilde{\lambda_{\mu}}}$.
\end{proof}

\subsection{Computing the divergence from the uniform distribution} \label{sec: divergence phi n}
Let $\vec{u}_{[n]}$ be the normalized vector (in the $\ell_2$ norm) corresponding to the distribution over $\sym$, which is uniformly supported over the $n$-cycles. Let $\vec{u}$ represent the uniform distribution over elements of $\sym$. Their inner product is given by $\braket{\vec{u}_{[n]}}{\vec{u}} = \frac{1}{\sqrt{n}}$. For some normalized distribution vector $\vec{u'}$, the quantity $\abs{\frac{1}{\sqrt{n}}-\braket{\vec{u}_{[n]}}{\vec{u'}}}$ serves as a measure of the distance between $\vec{u}$ and $\vec{u'}$.
In the quantum walk setting, we can use the state $\ket{\psi_{[n]}}$ in place of $\vec{u}_{[n]}$. We aim to determine the norm of the projection of $\ket{\psi_{(n)}}$ onto $\bw^t \ket{\psi_{\mathbb{e}}}$, as discussed earlier. In this section, we set out to explicitly determine this projection.

\begin{lemma}[projection lemma]\label{lm: projections}
    Given $\ket{\rho_{\mu,i,j}},A, B, S, f,\ket{\psi_{\mathbb{e}}}$ and $\ket{\psi_{[n]}}$ as defined earlier  the following holds for the  operator $\bw$:
    \begin{enumerate}
        \item     
        $\bra{\rho_{\mu,i,j}}A^\dagger\ket{\psi_{\mathbb{e}}} = \delta_{ij} $
        \item $\bra{\rho_{\mu,i,j}}B^\dagger\ket{\psi_{\mathbb{e}}} = \frac{\delta_{ij}\chi_\mu([\sigma])}{\dim \rho_\mu}$
        \item $ \bra{\psi_{[n]}}A\ket{\rho_{\mu,i,j}} =  \frac{\sqrt{(n-1)!}\delta_{ij}\chi_{\mu}([n])}{\dim \rho_{\mu}}$
        \item $\bra{\psi_{[n]}}B\ket{\rho_{\mu,i,j}} = \frac{\delta_{ij}}{d\sqrt{(n-1)!}\dim \rho_\mu}\left(\sum_{[\tau]}\abs{[\tau]}\Upsilon([\tau])\chi_{\mu}([\tau])\right)$ where $\Upsilon(g) = \sum_{h \in [n]}f(g^{-1}h)$.
    \end{enumerate}
where the last sum is over the conjugacy classes of $\sym$.
\end{lemma}

\begin{proof}
We only prove the relations $2-4$ as the first one is easy to check. In what follows we denote $d = \abs{S} = {n \choose 2}$.
    Recall, $B = \sum_{g \in \sym}\ket{\phi'_g}\bra{g}$, where $\ket{\phi'_g} = \frac{1}{\sqrt{d}
    }\sum_{s \in S}\ket{gs^{-1},g}$. Then,
    \begin{align*}
        B^\dagger \ket{\psi_{\mathbb{e}}} = \sum_{g \in G}\braket{\phi'_g}{\psi_{\mathbb{e}}}\ket{g} = \frac{1}{d}\sum_{s \in S}\ket{s}
    \end{align*}
    Hence,
    \begin{align*}
        \bra{\rho_{\mu,i,j}}B^\dagger\ket{\psi_\mathbb{e}} & = \frac{1}{d}\sum_{s \in S}\rho_{\mu}(s)[i,j]^* = \frac{1}{d}\sum_{s \in S}\rho_{\mu}(s)[i,j] = \frac{1}{d}\sum_{g \in G}X_S(g)\rho_\mu(g)[i,j]\\ & = \frac{\hat{X}_S(\rho_\mu)[i,j]}{d} = \frac{\delta_{ij}\chi_{\mu}([\sigma])}{\dim \rho_\mu}
    \end{align*}
Similarly,
    \begin{align*}
        \bra{\psi_{[n]}}A &= \frac{1}{\sqrt{(n-1)!}}\sum_{g \in [n]}\braket{\phi_g}{\phi_g}\bra{g} = \frac{1}{\sqrt{(n-1)!}}\sum_{g \in [n]}\bra{g}
    \end{align*}
        \begin{align*}
        \bra{\psi_{[n]}}A\ket{\rho_{\mu,i,j}} &= \frac{1}{\sqrt{(n-1)!}}\sum_{g \in [n]}\bra{g}\ket{\rho_{\mu,i,j}} = \frac{1}{\sqrt{(n-1)!}}\sum_{g \in [n]}\rho_{\mu}(g)[i,j] \\ & = \frac{1}{\sqrt{(n-1)!}} \sum_{g \in G}X_{[n]}(g)\rho_\mu(g)[i,j]= \frac{\sqrt{(n-1)!}\delta_{ij}\chi_{\mu}([n])}{\dim \rho_{\mu}} 
    \end{align*}
and 
        \begin{align*}
        \bra{\psi_{[n]}}B &= \frac{1}{\sqrt{(n-1)!}}\sum_{h \in [n]}\sum_{g \in G}\braket{\phi_h}{\phi'_g}\bra{g} = \frac{1}{d\sqrt{(n-1)!}}\sum_{g \in G}\sum_{h \in [n]}f(g^{-1}h)\bra{g} 
    \end{align*}
           \begin{align*}
        \bra{\psi_{[n]}}B\ket{\rho_{\mu,i,j}} &=  \frac{1}{d\sqrt{(n-1)!}}\sum_{g \in G}\rho_{\mu}(g)[i,j]\sum_{h \in [n]}f(g^{-1}h) 
    \end{align*}
    Let,  $\Upsilon(g) = \sum_{h \in [n]}f(g^{-1}h)$. Since $f$ is a class function, $\Upsilon(tgt^{-1}) = \sum_{h \in [n]}f(tg^{-1}t^{-1}h) = \sum_{h \in [n]}f(g^{-1}t^{-1}ht) = \Upsilon(g) $. Hence, 
   $\Upsilon$ is a class function.
    Thus,
    \begin{align*}
        \bra{\psi_{[n]}}B\ket{\rho_{\mu,i,j}} &=  \frac{1}{d\sqrt{(n-1)!}}\hat{\Upsilon}(\rho_\mu)[i,j] = \frac{\delta_{ij}}{d\sqrt{(n-1)!}\dim \rho_\mu}\left(\sum_{[\tau]}\abs{[\tau]}\Upsilon([\tau])\chi_{\mu}([\tau])\right)
    \end{align*} 
    where the last sum is taken over the conjugacy classes of $\sym$, and $\hat{\Upsilon}$ denotes the Fourier transform of $\Upsilon$ (as defined in Section~\ref{sec: fouier trans}).
\end{proof}
Let, 
\begin{align}\label{eqn: sum gamma over tau}
    \gamma_{\mu} = \frac{1}{
    \dim \rho_\mu}\left(\sum_{[\tau]}\abs{[\tau]}\Upsilon([\tau])\chi_{\mu}([\tau])\right) = \frac{\Tilde{\gamma}_{ \mu}}{\dim \rho_\mu}
\end{align}
Deriving the exact expression for $\gamma_{\mu}$ is quite challenging. Instead, we will attempt to find asymptotic bounds for them, which will serve our intended purpose.

\begin{lemma}\label{lm: upsilon sigma}
If $S$ is the class of transpositions and $\Upsilon$ is as defined in Lemma~\ref{lm: projections},
\begin{align*}
    \Upsilon(g) = \begin{cases}
    \mbox{$l(n-l)$ if $g \in [(l,n-l)]$ and}\\
    \mbox{$0$ otherwise}
    \end{cases}
\end{align*}
\end{lemma}
\begin{proof}
The proof follows directly from an application of Fact \ref{fc: 2} and a straightforward counting argument.
\end{proof}
\noindent Let $\mathcal{H} = \{[(l,n-l)] \mid 1 \leq l \leq n\}$, and we identify the elements of $\mathcal{H}$ with $\tau_l$. In equation~\ref{eqn: sum gamma over tau}, we only need to consider the sum over $\mathcal{H}$. The following lemma is a consequence of Lemma~\ref{lm: upsilon sigma}. 

\begin{lemma}\label{lm: gamma mu}
    For any $\mu$, we have $\Tilde{\gamma}_\mu = n!\iota_\mu$, where
    \begin{align*}
       \iota_{\mu} = \begin{cases}
         \mbox{$\sum_{l=1}^{\lfloor n/2 \rfloor}\chi_\mu([\tau_l]) $ if $n$ is odd}\\
         \mbox{$\sum_{l=1}^{\lfloor n/2 \rfloor - 1}\chi_\mu([\tau_l]) + \frac{n!}{2}\chi_{\mu}([\tau_{n/2}])$ otherwise}
       \end{cases}   
    \end{align*}
\end{lemma}
\begin{proof}
    Note that $\abs{[\tau_l]} = {n \choose l}(l-1)!(n-l-1)! = \frac{n!}{l(n-l)}$ except when $n$ is even and $l = n/2$, in that case, $\abs{[\tau_{n/2}]} = \frac{2(n-1)!}{n}$.
\end{proof}
\noindent If $\mu \in \Xi_{n,k}$ we can be more specific.
\begin{corollary}
    If $\mu \in \Xi_{n,k}$ and $\tau_l \in \mathcal{H}$ then
\begin{align*}
       \Tilde{\gamma}_{\mu} = \begin{cases}
            \mbox{$\frac{1}{2}(-1)^{n-k-1}(n-2k+1)n!$\ \ \ if $n$ is odd}\\
            \mbox{$\frac{1}{2}(-1)^{n-k-1}(n-2k+\frac{4}{n^2})n!$ \ \ \ otherwise}
        \end{cases}
\end{align*}
\end{corollary}

\begin{proof}
Using Fact \ref{fc: 4} and Lemma \ref{lm: gamma mu}  we have,
    \begin{align*}
        \Tilde{\gamma}_{\mu} &=  \sum_{l \ge k}^{\lfloor n/2 \rfloor}{\abs{[\tau_l]}l(n-l)(-1)^{n-k-1}}  = \begin{cases}
            \mbox{$\frac{1}{2}(-1)^{n-k-1}(n-2k+1)n!$\ \ \ if $n$ is odd}\\
            \mbox{$\frac{1}{2}(-1)^{n-k-1}(n-2k+\frac{4}{n^2})n!$ \ \ \ otherwise}
        \end{cases}
    \end{align*}
\end{proof}



\ifx false
\begin{lemma}
If $\mu = (k,1^{n-k})$  and $[\tau_l]$ comes from the class with one $l$ cycle and another $n-l$ cycle,
\begin{align*}
       \sum_{[\tau_l]}\#(\sigma)\Upsilon_S([\tau_l])\chi_{\mu}([\tau_l]) = \begin{cases}
           \mbox{$\sum_{l=1}^{\lceil n/2 \rceil}(-1)^{n-l-1} {n \choose l}l!(n-l)!$ if $l \ge k$}\\ \mbox{$0$ otherwise}
       \end{cases}
\end{align*}
\end{lemma}
\begin{proof}
Note that $\#(\tau_l) = {n \choose l}(l-1)!(n-l-1)! = \frac{n!}{l(n-l)}$ except when $n$ is even and $l = n/2$, in that case, $\#(\tau_{n/2}) = \frac{2(n-1)!}{n}$.
Using the expression for $\chi_{(k,1,\ldots,1)}(\tau_l)$ we have,
    \begin{align*}
        \sum_{[\tau_l]}\#(\tau_l)\Upsilon_{[\sigma]}([\tau_l])\chi_{\mu}([\tau_l]) &=  \sum_{l \ge k}^{\lfloor n/2 \rfloor}{\#(\tau_l)l(n-l)(-1)^{n-k-1}} \\ &= \begin{cases}
            \mbox{$\frac{1}{2}(-1)^{n-k-1}(n-2k+1)n!$\ \ \ if $n$ is odd}\\
            \mbox{$\frac{1}{2}(-1)^{n-k-1}(n-2k+\frac{4}{n^2})n!$ \ \ \ otherwise}
        \end{cases}
    \end{align*}
\end{proof}

\begin{corollary}
Let $\ket{u} = \frac{1}{\sqrt{n!d}}\sum_{g \in G, s \in S}\ket{g,gs} $, the state corresponding to the uniform distribution. Then both
$\bra{u}A\ket{\rho_{\mu,i,j}}, \bra{u}B\ket{\rho_{\mu,i,j}}$ either $\frac{n!}{\sqrt{(n-1)!}}$ (resp. if $\rho_\mu$ is trivial) or $0$ (resp. if $\rho_{\mu}$ is non-trivial).
\end{corollary}
\begin{proof}
    Since for any constant function ($f(g) = c\ \forall g \in G$) their Fourier coefficients are $0$ unless the irrep. is non-trivial. More specifically, $\hat{f}(\rho) = \sum_{g \in G}f(g)\rho_{\mu}(g) = c\sum_{g \in G}\rho(g) = 0$ unless $\rho(g) = 1 \ \forall g \in G$. 
\end{proof}
\fi
In the next lemma, we show that the component $\bw_{A,B}$ can be ignored in the spectral decomposition of the walk operator. This greatly simplifies the analysis and allows us to focus on the action of $\bw$ on the space orthogonal to the trivial singular values of $D$.

\begin{lemma}\label{lmm: null projecion}
Given $\bw_{A,B}^t, \ket{\psi_{[n]}},\ket{\psi_{\mathbb{e}}}$ as defined earlier $\bra{\psi_{[n]}}\bw_{A,B}^t\ket{\psi_{\mathbb{e}}} = 0$.
\end{lemma}
\begin{proof}
We have $\bw_{A,B}^t = \Pi_{+1} + (-1)^t \Pi_{-1} = \Pi_{\col(A)\cap \col(B)} + \Pi_{\ker(A)\cap \ker(B)} + (-1)^t\Pi_{\col(A)\cap \ker(B)} + (-1)^t\Pi_{\ker(A)\cap \col(B)}$. Since $\ket{\phi_{\mathbb{e}}} \in \col(A)$,
\begin{align*}
    \bw_{A,B}^t \ket{\phi_{\mathbb{e}}} = (\Pi_{\col(A)\cap \col(B)} + (-1)^t\Pi_{\col(A)\cap \ker(B)})\ket{\phi_{\mathbb{e}}} = \begin{cases} \mbox{$\ket{\phi_{\mathbb{e}}}$ if $t$ is even} \\ \mbox{$(2\Pi_B-I)\ket{\phi_{\mathbb{e}}}$ otherwise}
    \end{cases}
\end{align*}
In either case, we have $\bra{\phi_{[n]}}\bw_{A,B}^t \ket{\phi_{\mathbb{e}}} = 0$. Intuitively, this is because $\ket{\phi_{[n]}}$ has positive support only on the conjugacy class $[n]$, while $\bw_{A,B}^t \ket{\phi_{\mathbb{e}}}$ has support solely on $[\sigma]\cup \{\mathbb{e}\}$ in their first register, respectively.
\end{proof}
From the above lemma, it follows that $\bra{\phi_{[n]}}\bw^t \ket{\phi_{\mathbb{e}}} = \sum_{\mu \vdash n}\bra{\phi_{[n]}}\bw_\mu^t\ket{\phi_{\mathbb{e}}}$. Finally, we are ready to prove our main theorem.


\maintheorem*
\begin{proof}
Using the spectral idempotents in the decomposition of $\bw_\mu^t$ from Equation~\ref{eq: mu spectra}, we derive the following intermediate terms.

    \begin{align*}
        \alpha_1(\mu) = \bra{\phi_{[n]}} & (A\Pi_\mu A^\dagger +   B \Pi_\mu B^\dagger)\ket{\phi_{\mathbb{e}}}\\ & = \frac{\dim \rho_\mu}{n!}\sum_{1\le i,j\le \dim \rho_\mu}\left(\bra{\phi_{[n]}}A\ket{\rho_{\mu,i,j}}\bra{\rho_{\mu,i,j}}A^\dagger\ket{\phi_{\mathbb{e}}} + \bra{\phi_{[n]}}B\ket{\rho_{\mu,i,j}}\bra{\rho_{\mu,i,j}}B^\dagger\ket{\phi_{\mathbb{e}}} \right)\\ & = \frac{\dim \rho_\mu}{n!} \sum_{1\le i,j\le \dim \rho_\mu} \delta_{ij}\left(\frac{\sqrt{(n-1)!}\chi_\mu([n])}{\dim \rho_\mu}+\frac{\chi_\mu([\sigma])\tilde{\gamma}_\mu }{d\sqrt{(n-1)!}\dim^2 \rho_\mu}\right)\\
&= \frac{\dim \rho_\mu}{n!}\left(\sqrt{(n-1)!}\chi_\mu([n])+ \frac{\chi_\mu([\sigma])\tilde{\gamma}_{\mu}}{d\sqrt{(n-1)!}\dim \rho_\mu}\right)
    \end{align*}
   Similarly, we have 
   \begin{align*}
       \alpha_2(\mu) = \bra{\phi_{[n]}}  A\Pi_\mu B^\dagger\ket{\phi_{\mathbb{e}}}  = \frac{\chi_\mu([n])\chi_\mu([\sigma])}{n\sqrt{(n-1)!}}
   \end{align*}
   and 
      \begin{align*}
       \alpha_3(\mu) = \bra{\phi_{[n]}}  B\Pi_\mu A^\dagger\ket{\phi_{\mathbb{e}}}  = \frac{\Tilde{\gamma}_{\mu}}{d n!\sqrt{(n-1)!}}
   \end{align*}
We can rewrite $\alpha_1(\mu)$ as:
\begin{align*}
    \alpha_1(\mu) = \frac{\dim \rho_\mu}{\chi_\mu([\sigma])}\alpha_2(\mu) + \chi_\mu([\sigma])\alpha_3(\mu)
\end{align*}
If $\mu \in \Xi_{n,k}$ where $\mu = (k,1,\ldots,1)$ then,
\begin{align*}
    \alpha_2(\mu) &= \frac{(-1)^{n-k}{n-1 \choose k-1}(n-2k+1)}{m} \\
    & \mbox{and}\\
    \alpha_3(\mu) & \approxeq \frac{b_{n,k}}{m}\\ & \mbox{thus}\\
        \alpha_1(\mu) &= \frac{(n-1)(-1)^{n-k+1}}{m}{n-1 \choose k-1}\left(1 + \frac{(-1)^{n-k}b_{n,k}(n-2k+1)}{(n-1)^2}\right)
\end{align*}
where $m = n(n-1)\sqrt{(n-1)!}$ and 
\begin{align*}
b_{n,k} =    \begin{cases}
            \mbox{$(-1)^{k}(n-2k)$\ \ \ if $n$ is odd}\\
            \mbox{$(-1)^{k-1}(n-2k)$ \ \ \ otherwise}
        \end{cases} 
\end{align*}
Else if $\mu \in \Xi_n \setminus \Xi_{n,k}$ 
then:

\begin{align*}
    \alpha_2(\mu) & = 0\ \mbox{\ , since $\chi_\mu([n]) = 0$} \\ & \mbox{and} \\
     \alpha_3(\mu) & = \frac{2\iota_\mu}{m}
\end{align*}
Now, 
  \begin{align*}
\bra{\phi_{[n]}}\bw_\mu^t\ket{\phi_{\mathbb{e}}} &= \frac{1}{2(1-\lambda^2_\mu)}\left(\alpha_1(\mu) \cos (2\theta_\mu t) - s_\mu \alpha_2(\mu)\cos (2\theta_\mu(t-1/2)) -s_\mu \alpha_3(\mu) \cos(2\theta_\mu(t+1/2)) \right) 
   \end{align*}
Let $c_1 = \cos (2\theta_\mu t), c_2 = \cos (2\theta_\mu(t-1/2))$ and $c_3 = \cos (2\theta_\mu(t+1/2))$.
To compute $\bra{\phi_{[n]}}\bw_\mu^t\ket{\phi_{\mathbb{e}}}$ we only need to sum over $\mu \in \Xi_n$. Thus, 

\begin{align}\label{eq: w using alpha}
\nonumber \sum_{\mu \in \Xi_n}\bra{\phi_{[n]}}\bw_\mu^t\ket{\phi_{\mathbb{e}}} & \le \frac{1}{2m}\sum_{\mu \in \Xi_n} \frac{1}{1-\lambda_\mu^2}\left(\alpha_1(\mu) - s_\mu \alpha_2(\mu)c_2 - s_\mu\alpha_3(\mu)c_3\right)\\
\nonumber & = \frac{1}{2m}\sum_{\mu \in \Xi_n \setminus \Xi_{n,k}} \frac{1}{1-\lambda_\mu^2}\left(2(\chi_\mu([\sigma])-s_\mu c_3)\iota_\mu\right) \\
& + \frac{1}{2m}\sum_{\mu \in \Xi_{n,k}}\frac{1}{1-\lambda_\mu^2}P(n,k,\mu)
\end{align}
where, 
\begin{align*}
   P(n,k,\mu) &= \frac{(n-1)(-1)^{n-k+1}}{m}{n-1 \choose k-1}\left(1 + \frac{(-1)^{n-k}b_{n,k}(n-2k+1)}{(n-1)^2}\right)\\ & -s_\mu{n-1 \choose k-1}(-1)^{n-k}(n-2k+1)c_2-s_\mu b_{n,k}c_3
\end{align*}
Next, we use Lemmas~\ref{lm: mu tau l} through \ref{lm: lambda mu bound}:
\begin{enumerate}
\item Since each $\chi_\mu([\tau_l]) \in \{-1,0,1\}$ (by Lemma~\ref{lm: mu tau l}), we have $\abs{\iota_\mu} \le n/2$.
\item From Lemma~\ref{lm: bound chi mu sigma}, we have $\chi_\mu([\sigma]) = O(n^{6.5} \beta^{n})$ (for any $\beta > 81/16$).
\item Additionally, Lemma~\ref{lm: lambda mu bound} yields $\frac{1}{1-\lambda_\mu^2} \le \frac{n}{2}$ for $n \ge 2$ and for any $\mu \in \Xi_n$.
\end{enumerate}
Thus for any $\mu \in \Xi_n \setminus \Xi_{n,k}$,
\begin{align}\label{eq: mu general}
     \frac{1}{1-\lambda_\mu^2}\left(2(\chi_\mu([\sigma])-s_\mu c_3)\iota_\mu\right)  = O(n^{8.5} \beta^n) 
\end{align}
We also have,
\begin{align*}
    P(n,k,\mu) = O\left(n^2 {n-1 \choose k-1}\right)
\end{align*}
for $\mu \in \Xi_{n,k}$.
Hence,
\begin{align}\label{eq: pnk}
    \frac{1}{1-\lambda_\mu^2}P(n,k,\mu) = O(n^32^n)
\end{align}
Using Fact~\ref{fc: xi size} and substituting the expression for the left-hand side of Equations~\ref{eq: mu general} and \ref{eq: pnk} into equation~\ref{eq: w using alpha}, we obtain:

\begin{align*} \sum_{\mu \in \Xi_n}\bra{\phi_{[n]}}\bw_\mu^t\ket{\phi_{\mathbb{e}}} = O\left(\frac{n^{10}\beta^n}{\sqrt{n!}}\right)
\end{align*}
Thus,
\begin{align*}
\norm{\bra{\phi_{[n]}}\bw^t\ket{\phi_{\mathbb{e}}}} = O\left(\frac{n^{20}{\beta}^{2n}}{n!}\right)
\end{align*}
\end{proof}

\subsection{Computing the probability of observing $\ket{g,gs}$ for some $g \in [n]$}\label{sec: instan prob}
Here, we provide a brief discussion on why the above analysis fails (at least when applied directly) if we consider determining the probability of detecting a basis state $\ket{g,gs}$, where $g$ is an $n$-cycle.
The key to our analysis was the projection lemma (Lemma \ref{lm: projections}). However, in this case, we are not dealing solely with class functions, which implies explicit computation of the irreps that lack straightforward analytical expressions. More specifically, we wish to compute $\sum_{s \in S}\norm{\bra{g,gs}\bw^t\ket{\phi_{\mathbb{e}}}} $. Additionally, due to our choice of the starting state, the individual probabilities $\norm{\bra{g,gs}\bw^t\ket{\phi_{\mathbb{e}}}}$ do not depend on $s$.
Now,
\ifx false
\begin{align*}
    \bra{g,gs}A\ket{\rho_{\mu,i,j}} &= \frac{1}{\sqrt{d}}\rho_{\mu}(g)[i,j]\\
    \bra{g,gs}B &= \frac{1}{\sqrt{d}}\sum_{g' \in G}f(g'^{-1}g)\bra{g'}
\end{align*}
\fi
\begin{align*}
    \bra{g,gs}B\ket{\rho_{\mu,i,j}} = \frac{1}{\sqrt{d}}\sum_{h \in G}f_S(g^{-1}h)\rho_{\mu}(h)[i,j]
\end{align*}
However, the function $f_{g}(h) = f(g^{-1}h)$ is not a class function, as can be easily seen. Thus, we cannot use the projection lemma in the manner we did earlier without knowing the irreps explicitly.

\bibliographystyle{splncs04}
\bibliography{ref}
\end{document}